\documentclass[conference,letterpaper]{IEEEtran}

\usepackage{authblk}

\usepackage[utf8]{inputenc} 
\usepackage[T1]{fontenc}
\usepackage{url}
\usepackage[cmex10]{amsmath} 

\usepackage{mleftright}       
\mleftright                   

\usepackage{graphicx}         
\usepackage{booktabs}         

\usepackage{amsmath,amssymb,amsfonts,amsthm}
\usepackage{enumitem}
\usepackage{mathtools}
\usepackage{algorithmic}
\usepackage{textcomp}
\usepackage[table,xcdraw]{xcolor}
\usepackage{tikz}

\usepackage{subcaption}
\usepackage{multicol,multirow}
\usepackage{comment}

\usepackage{nopageno}

\usepackage[style=ieee,citestyle=numeric-comp,sorting=nty]{biblatex}
\newtheorem{theorem}{Theorem}

\newtheorem{proposition}{Proposition}
\newtheorem{example}{Example}
\newtheorem{lemma}{Lemma}
\newtheorem{definition}{Definition}

\newtheorem{corollary}{Corollary}

\newcommand{\bit}{\begin{itemize}}
	\newcommand{\eit}{\end{itemize}}
\newcommand{\bcor}{\begin{cor}}
	\newcommand{\ecor}{\end{cor}}
\newcommand{\beq}{\begin{equation}}
	\newcommand{\eeq}{\end{equation}}
\newcommand{\beqn}{\begin{equation}}
	\newcommand{\eeqn}{\end{equation}}
\newcommand{\bea}{\begin{eqnarray}}
	\newcommand{\eea}{\end{eqnarray}}
\newcommand{\bean}{\begin{eqnarray*}}
	\newcommand{\eean}{\end{eqnarray*}}
\newcommand{\ben}{\begin{enumerate}}
	\newcommand{\een}{\end{enumerate}}

\DeclareMathOperator*{\smallplus}{\scalerel*{+}{\textstyle\sum}}
\usepackage{scalerel}

\newcommand{\bF}{\mathbb{F}}

\newcommand{\bN}{\mathbb{N}}

\newcommand{\bR}{\mathbb{R}}

\newcommand{\bZ}{\mathbb{Z}}

\newcommand{\cA}{\mathcal{A}}
\newcommand{\cB}{\mathcal{B}}
\newcommand{\cC}{\mathcal{C}}

\newcommand{\cF}{\mathcal{F}}

\newcommand{\cH}{\mathcal{H}}

\newcommand{\cS}{\mathcal{S}}

\newcommand{\cU}{\mathcal{U}}
\newcommand{\cV}{\mathcal{V}}

\newcommand{\boldb}{\mathbf{b}}
\newcommand{\boldc}{\mathbf{c}}

\newcommand{\bolds}{\mathbf{s}}

\newcommand{\boldu}{\mathbf{u}}
\newcommand{\boldv}{\mathbf{v}}
\newcommand{\boldw}{\mathbf{w}}
\newcommand{\boldx}{\mathbf{x}}
\newcommand{\boldy}{\mathbf{y}}

\DeclarePairedDelimiter{\floor}{\lfloor}{\rfloor} 

\DeclareSymbolFont{bbold}{U}{bbold}{m}{n}
\DeclareSymbolFontAlphabet{\mathbbold}{bbold}
\newcommand{\1}{\mathbbold{1}}

\excludecomment{appendix}

\begin{document}
\title{Non-Binary Covering Codes \\ for Low-Access Computations} 


	\author{\textbf{Vinayak Ramkumar}$^{\dagger}$, \textbf{Netanel Raviv}$^\star$, and \textbf{Itzhak Tamo}$^\dagger$\\
		$^\dagger$Department of Electrical Engineering--Systems, Tel Aviv University, Tel Aviv, Israel\\
		$^\star$Department of Computer Science and Engineering, Washington University in St. Louis, St. Louis, MO, USA\\
  \texttt{vinram93@gmail.com}, \texttt{netanel.raviv@wustl.edu}, \texttt{zactamo@gmail.com}  
  \thanks{Parts of this work were done when V.~Ramkumar was a visiting researcher in the Department of Computer Science and Engineering, Washington University in St. Louis. The work of V.~Ramkumar and I.~Tamo was supported by the European Research Council (ERC grant number 852953). }
 }
\IEEEoverridecommandlockouts



\maketitle


\begin{abstract}
 Given a real dataset and a computation family, we wish to encode and store 
 the dataset in a distributed system so that any computation from the family can be performed by accessing a small number of nodes. In this work, we focus on the families of 
 linear computations where the coefficients are restricted to a finite set of real values. 
 For two-valued computations, a recent work presented a scheme that gives good feasible points on the access-redundancy tradeoff. This scheme is based on binary covering codes having a certain closure property. In a follow-up work, this scheme was extended to all finite coefficient sets, using a new additive-combinatorics notion called coefficient complexity. In the present paper, we explore non-binary covering codes and develop schemes that outperform the state-of-the-art for some coefficient sets. We provide a more general coefficient complexity definition and show its applicability to the access-redundancy tradeoff. 

\end{abstract}


\pagestyle{plain}

\section{Introduction}
Distributed storage and computation allow entities with limited resources (which we will refer to as the user) to perform computations over large datasets. In such setups, the user first stores its data in the nodes of the distributed system. At a later time, when the user wants to perform a computation over this distributed data, it queries the nodes and combines their answers to get the required computation result. 

Coded computation~\cite{TandonLDK17, yu2017polynomial,yu2019lagrange} deals with using coding theoretic techniques to address issues like stragglers, privacy, and error-resiliency in the distributed system. Typically, in coded computation settings, \textit{all} the nodes storing fractions of the data are contacted to perform any computation.
Such computation tasks could lead to excessive use of system resources, and in many cases, it is beneficial to reduce the number of nodes to be contacted preemptively. Accessing fewer nodes allows the remaining nodes to cater to a different computation or execute other tasks. 

Therefore, given a dataset and a family of computations, the user wishes to encode and store the dataset in a distributed system so that any computation from this family can be performed by accessing a small number of nodes.  
At storage time, the user encodes the data in a manner that will be later useful in reducing the number of nodes to be accessed to perform any computation in the family. Optimizing such system design gives rise to an access-redundancy tradeoff, where a smaller access requires more redundancy in the system and vice versa. At one end of the tradeoff, the data is stored without any redundancy, and computations require access to all nodes (minimum storage maximum access). At the other extreme, every possible computation result is stored in a separate node and hence just one node needs to be accessed to perform the computation (maximum storage minimum access). The challenge in this line of research is to come up with coding schemes that give good feasible intermediate points for function families of interest.

It is a well-known result that the access-redundancy tradeoff for linear computations over finite fields is equivalent to the problem of constructing linear codes with a small covering radius \cite{cohen1997covering}.  The current paper investigates access-redundancy tradeoffs in quantized linear computations over real numbers. Linear real-valued computations are common in many applications, in particular as part of machine learning inference.
Specifically, we focus on linear computations with quantized coefficients over real data, i.e., coefficients are restricted to a predetermined set of real values (such as~$\pm1$), and no restriction is imposed over the data.

Recently, quantized linear computations have attracted much interest due to applications in robust, efficient, and private machine learning models~\cite{qin2021towards,jia2020efficient,raviv2021enhancing,raviv2022information,deng2023private}. 
The access-redundancy tradeoff for quantized linear computations has been studied in the past.  For linear computations with $\{0,1\}$ coefficients, binary covering codes were used in \cite{ho1998partial} to get attractive points on the access-redundancy tradeoff. 
Several low-access schemes for two-valued computations were presented in~\cite{RamkumarRT23_ISIT} using binary covering codes having a certain closure property, and these schemes outperformed the results in \cite{ho1998partial}. More recently, a scheme was presented in \cite{RamkumarRT23_Allerton} for general coefficient sets, in which the storage remains identical irrespective of the coefficient set and the access parameter scales according to a newly introduced additive-combinatorics property called \textit{coefficient complexity}.

 The contributions of the current paper include the following. We propose a general method using covering codes over $\mathbb{F}_p$, where $p>2$ is a prime number, to come up with feasible points on the access-redundancy tradeoff for coefficient sets of size $p$ whose elements are in an arithmetic progression. For $p=3$, we construct several explicit low-access schemes by this method, and show that they outperform the previously known results for the common coefficient family $\{0,\pm1\}$. We present a general definition of coefficient complexity so that low-access schemes based on covering codes over $\mathbb{F}_p$ can be used to come up with low-access schemes for any finite coefficient set, thus extending a previous notion that exclusively pertained to~$\bF_2$. 


The rest of the paper is organized as follows. In Section~\ref{sec:background}, we provide the problem statement, give a brief overview of covering codes, and summarize previous results. Section~\ref{sec:non_binary} presents a low-access computation scheme based on non-binary covering codes and Section~\ref{sec:three_valued} specializes to the $p=3$ case. 

\section{Background} \label{sec:background}
We begin this section by defining the access-redundancy tradeoff problem for quantized linear computations, then give a brief introduction to covering codes and discuss prior results. 

For sets~$\cA,\cB\subseteq\bR$ we employ the standard set addition notation ~$\cA+\cB\triangleq\{a+b \mid a\in \cA,b\in \cB\}$. For the addition of multiple sets~$\cA_i$ we use $\smallplus_{i}\cA_i$. All vectors in this paper are row vectors, and~$\1$ is the all-ones row vector.

\subsection{Problem Statement}

Given data~$\boldx\in\bR^k$, it is encoded to~$\boldy\in\bR^n$, and stored in a distributed system with $n$ nodes such that each node stores one~$y_i\in\bR$.  We assume that the encoding uses a systematic linear code, i.e., every~$y_i$ is a linear combination of entries of~$\boldx$  and there is a node storing each $x_i$ in uncoded form.  Systematic encoding ensures it is possible to perform other tasks over the same stored data. A computation family of interest~$\cF$ is known at the encoding time. The encoding mechanism must be such that any~$f\in\cF$ can be computed by accessing at most~$\ell$ nodes, where~$\ell$ is the \textit{access parameter} of the system. The specific~$f\in\cF$ to be computed is not known at the time of encoding.
For a given $k$, it is preferable to minimize both storage~$n$ and access~$\ell$, but a clear tradeoff exists between these quantities.

This paper focuses on quantized linear computations over the real numbers.
For finite set $\cA \subseteq \bR$ (which will be referred to as the coefficient set), we define quantized linear computation family $$\cF_{\cA} \triangleq \{f(\boldx)=\boldw\boldx^\intercal\vert \boldw \in \cA^k\}.$$  
Given a coefficient set~$\cA$, we wish to develop an $\cA$-protocol 
which includes a storage (encoding) scheme, an access scheme that identifies~$\ell$ nodes to be contacted for computing a given $f \in \cF_{\cA}$, and a method to combine their responses to obtain $f(\boldx)$.
To understand the relation between~$n,k$, and~$\ell$ of $\cA$-protocols, we focus on the redundancy ratio~$n/k$ and the access ratio~$\ell/k$ in the asymptotic regime. 
We say that the pair~$(\alpha,\beta)$ is~$\cA$-feasible if there is an infinite family of $\cA$-protocols having parameters~$\{ (k_i,n_i,\ell_i) \}_{i\ge 0}$  (with~$k_i$'s strictly increasing) satisfying~$\lim_{i\to \infty}(n_i/k_i,\ell_i/k_i)=(\alpha,\beta)$. 
In \cite{RamkumarRT23_Allerton}, a protocol is called a \textit{universal protocol}  if the storage scheme does not depend on the coefficient set, 
i.e., all possible coefficient sets can be handled (with potentially different access ratios).

\subsection{Covering Codes}\label{section:coverigncodes} The covering radius~$r$  of a code~$\cC$ of block length~$m$ over an alphabet~$\Sigma$ is defined as $r \triangleq \min\{ u \mid \cup_{\boldc\in\cC}B_H(\boldc,u)=\Sigma^m\},$
where $B_H(\boldc,u)$ is the Hamming ball of radius $u$ centered at~$\boldc$.
Many of the best-known covering codes in the literature \cite{cohen1997covering} are obtained through a framework called \textit{amalgamated direct sum} described below. Let $\cC_z^{(i)}$ be the subset of~$\cC$ containing all codewords with~$i$-th coordinate equal to~$z$. 
The norm~$N^{(i)}$ is given by
$N^{(i)} = \max_{\boldv\in \Sigma^n}\left\{ \sum_{z\in\Sigma}d_H(\boldv,\cC_z^{(i)}) \right\}$, where $d_H$ denotes the Hamming distance and $d_H(\boldv,\cS)=\min_{\bolds\in\cS}d_H(\boldv,\bolds)$.
  The~$i$-th coordinate of~$\cC$ is said to be \textit{acceptable}, and~$\cC$ is said to be \textit{normal}, if~$N^{(i)}\le (r+1)|\Sigma|-1$ for some $i$. Let~$\cU,\cV$ be codes over an alphabet~$\Sigma$ having lengths~$m_1,m_2$ and covering radii~$r_1,r_1$. Suppose the last coordinate of~$\cU$ and the first coordinate of~$\cV$ are acceptable, and~$\cU_{z}^{(m_1)},\cV_{z}^{(1)}$ are not empty~for all~$z\in\Sigma$. It is known \cite{lobstein1989normal,graham1985covering,cohen1986further} that under these conditions, the amalgamated direct sum  $\cU\dotplus\cV\triangleq \bigcup _{z\in\Sigma}\{ (\boldu,z,\boldv) \mid  (\boldu,z)\in\cU,(z,\boldv)\in \cV\}$
	results in a code of length~$m_1+m_2-1$ and covering radius at most~$r_1+r_2$.

\subsection{Prior Results}
Here we give a summary of the previously known results on access-redundancy tradeoffs in quantized linear computations.   
It is shown in \cite{ho1998partial} that if there is a code $\cC$ over $\mathbb{F}_2$ of length~$m$ and covering radius $r$, then the pair $(\frac{m+|\cC|}{m},\frac{r+1}{m})$ is $\{0,1\}$-feasible.  This protocol involves partitioning~$\boldx$ into small batches of size $m$ and the non-systematic symbols are obtained by encoding each batch with a~$\{0,1\}$-matrix whose columns are all the codewords of $\cC$. 

This scheme was recently improved in \cite{RamkumarRT23_ISIT} using binary codes having a certain closure property. Let $\cC$ be a code over~$\bF_2$ and define~$\hat{\cC}$ to be a subset of $\cC$ containing exactly one of~$\{ \boldc,\boldc\oplus \1 \}$ for every~$\boldc\in \cC$; if the code $\cC$ over $\mathbb{F}_2$ has length $m$ and covering radius $r$, then the pair $(\frac{m+|\hat{\cC}|}{m},\frac{r+1}{m})$ is $\{\pm 1\}$-feasible. The columns of the $\{\pm1\}$-matrix used for encoding correspond to all the codewords that belong to $\hat{\cC}$. If the binary code~$\cC$ is closed under complement (i.e.,~$\boldc\in\cC$ if and only if~$\boldc \oplus \1 \in\cC$), then $|\hat{\cC}|=|\cC|/2$. Codes that are closed under complement thus result in savings of almost half the storage cost and this is the crux of the improvement in \cite{RamkumarRT23_ISIT}. Several low-access are obtained in \cite{RamkumarRT23_ISIT} by constructing binary covering codes having this closure property.  Additionally, it is shown that all two-valued computations have identical feasible pairs.  

More recently~\cite{RamkumarRT23_Allerton}, a universal protocol was presented based on a new additive-combinatorics notion. For any finite set $\cA \subset \bR$, the complexity of $\cA$, denoted~$C(\cA)$, is the smallest positive integer~$\theta$ such that there exists $2\theta$ real numbers $a_1, \dots,a_{\theta}, b_1,\dots,b_{\theta}$  (not necessarily distinct) satisfying $\cA \subseteq
	\smallplus_{i=1}^{\theta} \{a_i,b_i\}.$ 
 If there exists an $\{\pm1\}$-protocol which gives feasible pairs $(\alpha,\beta)$, then for all finite sets $\cA \subset \bR$,  it is shown in \cite{RamkumarRT23_Allerton} that $(\alpha,C(\cA)\beta)$ is $\cA$-feasible using the same encoding scheme. 
 
\section{Schemes based on Non-Binary Codes } \label{sec:non_binary}
As discussed previously, the protocols in \cite{RamkumarRT23_ISIT} were based on \textit{binary} covering codes. In this work, we explore the possibility of getting better feasible pairs using \textit{non-binary} codes. 
The following result is a direct consequence of the definition of covering radius; nevertheless, we provide a proof for completeness. 
\begin{theorem} \label{thm:general}
	If there exists a (not necessarily linear) code~$\cC$ of  length~$m$ and covering radius~$r$ over an alphabet $\Sigma$, then the pair~$(\frac{m+|\cC|}{m},\frac{r+1}{m})$ is~$\cA$-feasible, where $\cA \subset \bR$ is any set of size $|\Sigma|$. 
\end{theorem}  
\begin{proof}
    Let $\Sigma=\{s_1,\dots,s_p\}$ and $\cA=\{a_1,\dots,a_p\}$. Consider the representation $s_i \mapsto a_i$ for all $i \in [p]$. Let~$B \in  \cA^{m \times |\cC|}$~ be the matrix whose columns are all codewords of~$\cC$ in the above representation. 
Construct the $m\times(m+|\cC|)$ matrix  $M= \begin{bmatrix}
   I \mid B 
\end{bmatrix}$, where~$I \in \{0,1\}^{m \times m}$ is the  identity matrix.  

Consider $k=tm$, where~$t\in\bN$.  Partition~$\boldx\in\bR^{tm}$ to~$t$ parts each of size~$m$, denoted by ~$\boldx_1,\ldots,\boldx_t$. Then, encode each part into $\boldy_i=\boldx_iM$. Let~$\boldy=(\boldy_1,\ldots,\boldy_t)\in\bR^{t(m+|\cC|)}$ and distribute it to $n=t(m+|\cC|)$ nodes for storage. Clearly, the resultant redundancy ratio is~$\frac{m+|\cC|}{m}$.

The user can then compute~$\boldw\boldx^\intercal$ for any~$\boldw\in \cA^k$ as follows. Partition $\boldw$ into $t$ vectors of length~$m$ each, i.e., ~$\boldw=(\boldw_1,\ldots,\boldw_t)$.
 For each~$i\in[t]$, there exists a vector~$\boldb_i \in \cA^m$ such that $\boldw_i$ and $\boldb_i$ differ on at most $r$ entries and  $\boldb_i^\intercal$ is a column of $B$. This follows since $\cC$ has a covering radius $r$.  

 To compute~$\boldw_i\boldx_i^\intercal$, access the node storing $\boldb_i\boldx_i^\intercal$ and at most~$r$ nodes storing systematic symbols $x_{i,j}$'s such that ~$b_{i,j} \ne w_{i,j}$. Clearly,  ~$\boldw_i\boldx_i^\intercal$ can be computed using a linear combination of~$\boldb_i\boldx_i^\intercal$ and those~$x_{i,j}$'s.
  In total, at most~$t(r+1)$ nodes are accessed, leading to an access ratio of $\frac{r+1}{m}$.
\end{proof}
The challenge is to come up with schemes that outperform the feasible pairs given by Theorem~\ref{thm:general}, as was done for the binary case in \cite{RamkumarRT23_ISIT} using the aforementioned closure property.  In this paper, we study schemes based on covering codes over~$\bF_p$, where $p>2$ is a prime number, with a particular focus on ternary codes. For any prime number $p>2$, let $$\cA_p \triangleq \Big\{0, \pm1,\dots, \pm \frac{p-1}{2}\Big\}.$$ 
Our focus on sets of the form~$\cA_p$ is inspired by the commonly used \textit{uniform quantization}~\cite{gallager2008principles}, which results in coefficient sets whose elements are in arithmetic progression.
Later, in Proposition~\ref{proposition:equivalence}, it will be shown that all coefficient sets of size~$p$ which form an arithmetic progression have the same feasible pairs as $\cA_p$. Furthermore, Theorem~\ref{thm:complexity} will show that protocols for $\cA_p$ can be used as building blocks for a universal protocol. Motivated by these, we now devise an $\cA_p$-protocol, which reduces the redundancy in storage by using codes over~$\bF_p$ in a certain way.

  To avoid confusion, we use the notation~$\bar{\cdot}$ to emphasize that a given element (scalar or vector) is over~$\bF_p$ (rather than over~$\bR$). We use covering codes over~$$\bF_p=\{\bar{0},\bar{1},\dots, \overline{(p-1)}\}$$ in their \textit{$\bR$-representation} defined as follows:
\begin{align} \label{eq:R_representation}
\bar{0}&\mapsto 0, 
\bar{1} \mapsto 1, ~\dots~ 
\overline{\Big(\textstyle\frac{p-1}{2}\Big)} \mapsto \textstyle\frac{p-1}{2}, \nonumber \\ \overline {\Big(\textstyle\frac{p+1}{2}\Big)} &\mapsto -\textstyle\frac{p-1}{2}, ~\dots~,\overline{(p-1)} \mapsto -1. 
\end{align}
In particular for $\bF_3$, we have $\bar{0} \mapsto 0$, $\bar{1} \mapsto 1$ and $\bar{2} \mapsto -1$. 
Consider any integer $i$ such that $ 1 \le i \le \frac{p-1}{2}$. Then, clearly  $\bar{i}=-\overline{(p-i)}$   over $\bF_p$ corresponds to
  $i=-(-i)$ over $\bR$. It follows that under the above representation, the \textit{negation} operation is homomorphic. This enables reduction of storage overhead by using only one of~$\{\pm \bar{\boldc}\}$ for storage, for every~$\bar{\boldc}\in\cC$, as described below.

\begin{definition} \label{def:hatC}
	For a code~$\cC$ over~$\bF_p$, let~$\cS\subseteq \cC$ be all the codewords of Hamming weight at most~$1$, and let~$\tilde{\cC}\subseteq \cC\setminus \cS$ which contains exactly one of~$\{ \pm \bar{\boldc} \}$, for every~$\bar{\boldc} \in \cC \setminus 
	\cS$. That is, first remove all codewords from~$\cC$ that are of Hamming weight at most~$1$. Then, if~$\{ \pm \bar{\boldc} \}$ remain for some~$\bar{\boldc}$, remove one of them, and denote the resulting set by~$\tilde{\cC}$. 
\end{definition}

 There are multiple ways of generating~$\tilde{\cC}$ from~$\cC$ which result in codes of identical sizes, and we fix one such way arbitrarily.
 \begin{example} \label{example:hamming}
 	Consider the~$[4,2]_3$ Hamming code $\cC$ over~$\bF_3$. The $2 \times 4$ matrix  below is a generator matrix of this code:
 	\begin{align*} 
 		\begin{pmatrix}
 			\bar{0} & \bar{1} & \bar{1} & \bar{1} \\
 			\bar{1} & \bar{0} & \bar{1} & \bar{2} 
 		\end{pmatrix}.
 	\end{align*}
 This code has the following $9$ codewords:  
 \bean 
 \{\bar{0}\bar{0}\bar{0}\bar{0}, \bar{1}\bar{0}\bar{1}\bar{2}, \bar{2}\bar{0}\bar{2}\bar{1}, \bar{0}\bar{1}\bar{1}\bar{1}, \bar{1}\bar{1}\bar{2}\bar{0}, \bar{2}\bar{1}\bar{0}\bar{2}, \bar{0}\bar{2}\bar{2}\bar{2}, \bar{1}\bar{2}\bar{0}\bar{1}, \bar{2}\bar{2}\bar{1}\bar{0}\}. 
 \eean 
Here $\cS=\{\bar{0}\bar{0}\bar{0}\bar{0}\}$ and $|\cC \setminus \cS|=8$. For every $\bar{\boldc} \in \cC \setminus \cS$, we have $-\bar{\boldc} \in \cC \setminus \cS$. Therefore, $|\tilde{\cC}|=4$. It can also be seen that this code has a covering radius of $1$.  
 \end{example}
Based on the above definition, we state the following theorem which gives an improvement over Theorem~\ref{thm:general}. 

\begin{theorem}\label{theorem:Fp}
	Let $p>2$ be a prime number.  
	If there exists a (not necessarily linear) code~$\cC$ of length~$m$ and covering radius~$r$ over~$\bF_p$, then the pair~$(\frac{m+|\tilde{
			\cC}|}{m},\frac{r+1}{m})$ is~$\cA_p$-feasible.
\end{theorem}
\begin{proof}
Let $\tilde{c} \triangleq |\tilde{\cC}|$, and let~$B \in  \cA_p^{m \times \tilde{c}}$~ be the matrix whose columns are all vectors in~$\tilde{\cC}$ in their $\bR$-representation~\eqref{eq:R_representation}. That is, for each~$\bar{\boldc}\in\tilde{\cC}$, the matrix~$B$ contains as column the~$\bR$-representation~$\boldc$ of~$\bar{\boldc}$.
Construct the $m\times(m+ \tilde{c})$ matrix  $M= \begin{bmatrix}
   I \mid B 
\end{bmatrix}$, where~$I \in \{0,1\}^{m \times m}$ is the  identity matrix. 

Consider $k=tm$, where~$t\in\bN$.  Partition~$\boldx\in\bR^{tm}$ into~$t$ parts each of size~$m$, denoted by ~$\boldx_1,\ldots,\boldx_t$. Then, encode each part $\boldx_i$ into $\boldy_i=\boldx_iM$. Let~$\boldy=(\boldy_1,\ldots,\boldy_t)\in\bR^{t(m+\tilde{c})}$ and distribute $\boldy$ to the $n=t(m+\tilde{c})$ nodes for storage. The resultant redundancy ratio is~$\frac{m+\tilde{c}}{m}$.

The user can then compute~$\boldw\boldx^\intercal$ for any~$\boldw\in \cA_p^k$ as follows. Partition $\boldw$ into $t$ vectors of length~$m$ each, i.e., ~$\boldw=(\boldw_1,\ldots,\boldw_t)$.
 Since $\cC$ is of covering radius $r$, for each~$i\in[t]$, there is a codeword
 $\bar{\boldc}_i\in\cC$ of Hamming distance at most~$r$ from the~$\bF_p$-representation~$\bar{\boldw}_i$ of~$\boldw_i$, i.e., $d_H(\bar{\boldc}_i,\bar{\boldw}_i) \le r$. 
 The~$\bR$-representation of $\bar{\boldc}_i$, which is denoted by $\boldc_i$, thus differs from $\boldw_i$ on at most $r$ entries. Define $\Delta_{i} \triangleq \{j \in [m] \mid c_{i,j} \ne w_{i,j}\}$, the set of indices where $\boldw_i$ and $\boldc_i$ do not match. It follows from the above discussion that $|\Delta_{i}| \le r$, for all $i \in [t]$. We consider two cases based on whether or not there is a node storing~$\boldc_i\boldx_i^\intercal$.
	 \begin{enumerate}
	 	\item[(a)] If~$\bar{\boldc}_i\in\hat{\cC}$, then there is a node storing $\boldc_i\boldx_i^\intercal$. To compute $\boldw_i\boldx_i^\intercal$, the user accesses ~$\boldc_i\boldx_i^\intercal$, and at most $r$ systematic symbols $x_{i,j}$'s with $j \in \Delta_i$. Then,~$\boldw_i\boldx_i^\intercal$ can be computed as a linear combination of~$\boldc_i\boldx_i^\intercal$ and those~$x_{i,j}$'s as given below. 
   \begin{align*}
    &&  \boldc_i\boldx_i^\intercal + \sum_{j \in \Delta_i}(w_{i,j}-c_{i,j})x_{i,j} \hspace{3cm}\\
      &=& \sum_{j \in \Delta_i}c_{i,j}x_{i,j} + \sum_{j \notin \Delta_i}c_{i,j}x_{i,j}+\sum_{j \in \Delta_i}(w_{i,j}-c_{i,j})x_{i,j} \\ 
      &=&  \sum_{j \in [m]}w_{i,j}x_{i,j}=\boldw_i\boldx_i^\intercal. \hspace{4cm}
   \end{align*}
	 	\item[(b)] If~$\bar{\boldc}_i\notin \hat{\cC}$ then either~$w_H(\bar{\boldc}_i)\le 1$ or~$-\bar{\boldc}_i\in\hat{\cC}$. In the former case,
   by triangle inequality $w_H(\bar{\boldw}_i) \le r+1$ and hence~$\boldw_i\boldx_i^\intercal$ can be computed using at most $r+1$ systematic symbols. 
   In the latter case, access the node storing~$-\boldc_i\boldx^\intercal$ and negate it to get~$\boldc_i\boldx^\intercal$.
    At most~$r$ systematic nodes storing~$x_{i,j}$'s with $j \in \Delta_i$ also need to be accessed. Then, $\boldw_i\boldx_i^\intercal$ can be computed by linear combination of~$\boldc_i\boldx_i^\intercal$ and those~$x_{i,j}$'s as described above. 
	 \end{enumerate}	
  To conclude the computation, obtain $\boldw\boldx^\intercal=\sum_{i \in [t]} \boldw_i\boldx_i^\intercal$. 
  In total, at most~$t(r+1)$ nodes are accessed, leading to an access ratio of $\frac{r+1}{m}$.
\end{proof}
It can be seen that for coefficient set $\{0,\pm1\}$, using  the $[4,2]_3$ Hamming code in Example~\ref{example:hamming} as $\cC$, Theorem~\ref{thm:general}  gives feasible pair $(3.25, 0.5)$, whereas Theorem~\ref{theorem:Fp} gives a better feasible pair $(2,0.5)$.
The following corollary specializes Theorem~\ref{theorem:Fp} to the ternary case. 
\begin{corollary}\label{cor:F3}
	If there exists a (not necessarily linear) code~$\cC$ of length~$m$ and covering radius~$r$ over~$\bF_3$, then the pair~$(\frac{m+|\tilde{
			C}|}{p},\frac{r+1}{m})$ is~$\{0,\pm1\}$-feasible.
\end{corollary}
In the next section, we will present feasible pairs obtained using a few covering codes over $\mathbb{F}_3$ according to the above method and compare them against known results \cite{RamkumarRT23_ISIT,RamkumarRT23_Allerton}. We remark that codes over larger alphabets will lead to very large~$\alpha$, which may not be practically viable. 

Next, we focus on coming up with a universal protocol that uses non-binary covering codes as building blocks. 
For two-valued computations, it was shown in \cite{RamkumarRT23_ISIT} that the feasible pairs are identical irrespective of the coefficient values.   
However, for computations that require three or more coefficients such generality does not seem to hold in general, and yet, some restricted generality does hold. Specifically, the proposition below shows that all coefficient sets of size prime $p>2$ that form an arithmetic progression have identical feasible pairs. 
\begin{proposition}\label{proposition:equivalence}
	Let $p>2$ be a prime number and let $\cA$ be a set of size $p$ such that the elements of $\cA$ are in an arithmetic progression. 
	A pair~$(\alpha,\beta)$ is~$\cA$-feasible if and only if it is $\cA_p$-feasible.  Furthermore, data stored using a protocol for~$\cA_p$ can be used to retrieve~$\boldw\boldx^\intercal$ for any~$\boldw\in \cA^k$ by using at most one additional node. 
\end{proposition}
\begin{proof}
Let $\cA=\{a_1,\dots,a_p\}$ with $a_{i+1}-a_{i}=d > 0$, for all $i \in [p-1]$. 
    Given a protocol for~$\cA_p$, store~$\boldx$ according to it, with an additional node containing~$\1 \boldx^\intercal$ if not already present. To compute any given~$\boldw\in \cA^k$, consider~$\boldw'\in \cA_p^k$ defined as follows. 
If $w_j=a_{i}$, then $$w'_j \triangleq -\Big(\frac{p-1}{2}\Big)+i-1.$$ 
Using the given $\cA_p$-protocol for retrieving~$\boldw'\boldx^\intercal$, and accessing~$\1\boldx^\intercal$, the user can compute~
\begin{align*}
  && d\cdot \boldw'\boldx^\intercal+\Big(\frac{p-1}{2}d+a_1\Big)\cdot \1\boldx^\intercal \hspace{2cm} \\ 
&=&\sum_{i \in [p]}\sum_{j\vert w_j=a_i} \Big(-\Big(\frac{p-1}{2}\Big)d+(i-1)d+\frac{p-1}{2}d+a_1\Big) x_j  \\
&=&\sum_{i \in [p]}\sum_{j\vert w_j=a_i} \Big((i-1)d+a_1\Big) x_j \hspace{3.3cm} \\
&=&\sum_{i \in [p]}\sum_{j\vert w_j=a_i} a_ix_j  
= \boldw\boldx^\intercal. \hspace{4.2cm}  
\end{align*}
   
In the other direction, given a protocol for~$\cA$, store~$\boldx$ according to it, with an additional node containing~$\1\boldx^\intercal$ if not already present. To compute any given~$\boldw\in \cA_p^k$, consider~$\boldw'\in \cA^k$ defined as follows. 
    If $w_j=i$, then $$w'_j \triangleq a_{1}+\Big(\frac{p-1}{2}+i\Big)d.$$ 
Note that this is equivalent to $w'_j=a_{u}$, where $u=\frac{p-1}{2}+i+1$, and hence $\boldw \in \cA^k$. 
Using  the given $\cA$-protocol for retrieving~$\boldw'\boldx^\intercal$, and accessing~$\1\boldx^\intercal$, the user can compute
    \begin{align*}
  && 
  \frac{1}{d}\cdot \boldw'\boldx^\intercal-\Big(\frac{p-1}{2}+\frac{a_1}{d}\Big)\cdot \1\boldx^\intercal
  \hspace{1.5cm} \\ 
&=&\sum_{i \in \cA_p}\sum_{j\vert w_j=i} \Big(\frac{a_1}{d}+\frac{p-1}{2}+i-\frac{p-1}{2}-\frac{a_1}{d}\Big) x_j  \\
&=&\sum_{i \in \cA_p}\sum_{j\vert w_j=i} ix_j  
= \boldw\boldx^\intercal. \hspace{3.2cm}  
\end{align*}
At most one additional node is used for storage and access in both cases. Since feasibility is defined in an asymptotic way, it follows that the feasible pairs are identical for $\cA$ and $\cA_p$. The furthermore part follows directly from the first part of the proof. 
\end{proof}

Based on Proposition~\ref{proposition:equivalence}, we provide the following general definition of coefficient complexity. 
\begin{definition} \label{def:complexity}
For any finite set $\cA \subset \bR$ and any positive integer $p \ge 2$, the $p$-complexity~$C_p(\cA)$ of $\cA$ is the smallest positive integer~$\theta$ such that there exist $\theta$ many sets  $S_1, \dots,S_{\theta} \subset \bR$  (not necessarily distinct) satisfying  the following: 
\bit
\item   $\cA \subseteq
\smallplus_{i=1}^{\theta} S_i$ and 
\item  $|S_i|=p$ and the elements of  $S_i$ are in arithmetic progression, for all $i \in [\theta]$. 
\eit 

 \begin{example} \label{example:C3}
    For $\cA=\{0,1,2,3,4,5,6,7,8\}$ it is clear that $\cA \subseteq \{0,1,2\}+ \{0,3,6\}$. Since $\{0,1,2\}$ and $\{0,3,6\}$ are arithmetic progressions of size three,
    it follows that~$C_3(\cA)=2$. For $\cA=\{0,1,2,3,4,5,6,7,9\}$, we have $\cA \subseteq \{0,1,2\}+ \{0,3,6\}+\{0,9,18\}$, and hence $C_3(\cA) \le 3$.
\end{example}

\end{definition}
It can be easily verified that for $p=2$ Definition~\ref{def:complexity} specifies to~\cite[Def.~1]{RamkumarRT23_Allerton}.
Using arguments similar to those in~\cite{RamkumarRT23_Allerton}, the following upper and lower bounds can derived for $p$-complexity $C_p(\cA)$.
\begin{lemma} \label{lem:bound}
	For all finite sets $\cA \subset \bR$ of size at least two,   we have $\lceil \log_p(|\cA|) \rceil  \le C_p(\cA) \le |\cA|-1$. 
\end{lemma}
\begin{proof}
     Let $C_p(\cA)=\theta$. For any collection of sets  $S_1, \dots,S_{\theta} \subset \bR$	each of size $p$, it can be seen that
		\begin{equation*}
\left|\smallplus_{i=1}^{\theta} S_i \right|
 \le p^{\theta}.  
		\end{equation*}
Therefore, $|\cA| \le p^{\theta}$. The lower bound follows.   
 
To obtain the upper bound, without loss of generality, let $\cA=\{a_1,\dots,a_M\}$ with $a_1<\dots<a_M$.
Define $(M-1)$ many $p$-element sets in the following fashion: 
 $$S_i=\{a_i-a_M+(a_M-a_i)j \mid j \in \bZ, 0 \le j \le p-1\}$$
 for all $i\in [M-2]$ and
$$S_{M-1}=\{a_{M-1}+(a_M-a_{M-1})j \mid j \in \bZ, 0 \le j \le p-1\}.$$
Clearly, elements of $S_i$ are in arithmetic progression. 
Observe that $a_{i}-a_{M} \in S_i$ for all $i \in [M-2]$ and that $a_{M} \in S_{M-1}$. 
Then, for any $i \in [M-2]$, we have
$$a_i=(a_i-a_M+a_M) \in S_{i}+S_{M-1}.$$ Since $0 \in  S_i$ for all $i \in [M-2]$, it follows that $$a_i \in \smallplus_{j=1}^{M-1} S_j$$ for all $i \in [M-2]$. Also note that $\{a_{M-1},a_{M} \} \subseteq S_{M-1}$ results in $$\{a_{M-1},a_{M} \} \subseteq \smallplus_{j=1}^{M-1} S_j.$$ Thus, we have shown that $\cA \subseteq \smallplus_{j=1}^{M-1} S_j$, proving the upper bound. 
\end{proof}
The following theorem establishes the usefulness of our new $p$-complexity definition by providing a universal protocol. 
\begin{theorem} \label{thm:complexity}
Let $p>2$ be a prime number. If there exists an $\cA_p$-protocol with feasible pair $(\alpha,\beta)$, then $(\alpha,C_p(\cA)\beta)$ is $\cA$-feasible using the same encoding scheme, for all finite sets $\cA \subset \bR$. 
\end{theorem} 
\begin{proof}
    Let $C_p(\cA)=\theta$. By the definition of $p$-complexity, there exists $S_1, \dots,S_{\theta} \subset \bR$ 
satisfying
$$\cA \subseteq
\smallplus_{i=1}^{\theta} S_i$$ 
such that each $S_i$ is a set of size $p$ whose elements form an arithmetic progression. It follows that for given any $\boldw \in \cA^k$, 
one can find $\boldw^{(i)} \in S_i^k$, $i \in [\theta]$, such that
$$ \boldw =\sum_{i \in [\theta]}\boldw^{(i)}.$$ Therefore,
$$ \boldw\boldx^\intercal =\sum_{i \in [\theta]}\boldw^{(i)}\boldx^\intercal.$$
Given a protocol for~$\cA_p$, store~$\boldx$ according to it, with an additional node storing~$\1 \boldx^\intercal$ if not already present. Since every $S_i$ is an arithmetic progression, it follows from the proof of Proposition~\ref{proposition:equivalence} that for each $\boldw^{(i)}$ 
there exists a $\hat{\boldw}^{(i)} \in \cA_p$, such that $\boldw^{(i)}\boldx^\intercal$  can 
be computed using $\hat{\boldw}^{(i)}\boldx^\intercal$ 
and $\1 \boldx^\intercal$. 
Therefore, the total access required for computing $\boldw\boldx^\intercal$ is at most $\theta$ times the access requirement of the given $\cA_p$-protocol plus one additional node storing $\1 \boldx^\intercal$. The statement of the theorem follows due to the limit operation in the definition of feasible pairs.   
\end{proof}
From Lemma~\ref{lem:bound}, it follows that for any coefficient set $\cA$, the universal protocol given by Theorem~\ref{thm:complexity} has access ratio $(|\cA|-1)\beta$ in the worst-case. If $\cA$ is not a set of the highest complexity, then a smaller access ratio is possible.
Hence, further understanding of the complexity of various sets would be beneficial. 
 Similar to \cite{RamkumarRT23_Allerton}, it is possible to provide alternate definitions of $p$-complexity and prove results on $p$-complexity of arithmetic and geometric progressions. However, we do not delve into this in the current paper and it is left as future work.
\section{Three-valued Computations} \label{sec:three_valued}

In this section, we present explicit schemes by applying Corollary~\ref{cor:F3} to several ternary covering codes. 
It is shown that the resultant $\{0, \pm 1\}$-feasible pairs outperform the best-known  results~\cite{RamkumarRT23_Allerton,RamkumarRT23_ISIT}. For any code $\cC$, we define $\tilde{c} \triangleq |\tilde{\cC}|$, where $\tilde{\cC}$ is the subset of $\cC$ given in Definition~\ref{def:hatC}.

\paragraph{Entire space}
The entire space~$\cC=\bF_3^i$ is a covering code of size~$3^i$ and covering radius~$r=0$.  Clearly,~$\tilde{c}=\frac{3^i-2i-1}{2}$ and the pair~$\left\{ \left(\frac{3^i-1}{2i},\frac{1}{i} \right) \right\}_{i\ge 1}$ is $\{0,\pm 1\}$-feasible.

\paragraph{Repetition code}
It is readily verified that the repetition code~$\cC=\operatorname{span}_{\bF_3}\{\bar{\1}\}\subseteq \bF_3^i$ is a code of size~$3$ and covering radius~$r=\floor{\frac{2i}{3}}$. It also readily verified that~$\tilde{c}=1$, which gives rise to the $\{0,\pm 1\}$-feasible pairs~$\left\{ \left(\frac{i+1}{i}, \frac{1+\floor{\frac{2i}{3}}}{i}\right) \right\}_{i\ge 1}.$
By setting $k=i$, we get that the pair~$(1,2/3)$ is $\{0,\pm 1\}$-feasible.

\paragraph{Hamming code} 
The~$[4,2]_3$ Hamming code $\cC$ over~$\bF_3$ generated by
\begin{align*} 
G_\text{H} =\begin{pmatrix}
	\bar{0} & \bar{1} & \bar{1} & \bar{1} \\
	\bar{1} & \bar{0} & \bar{1} & \bar{2} 
\end{pmatrix},
\end{align*}
is a code of covering radius~$1$. 
As shown in Example~\ref{example:hamming}, clearly $\tilde{c}=4$.  
This gives rise to the $\{0,\pm 1\}$-feasible pair~$(2,1/2)$. 

We now use the amalgamated direct sum technique described in Section~\ref{section:coverigncodes} to get more feasible points. Notice that the~$[4,2]_3$ Hamming code is not normal, and hence cannot be amalgamated.
\paragraph{Expanded Hamming code and its amalgamations} 
We define the \textit{expanded} Hamming code over~$\bF_3$ with generator matrix 
\begin{align*} 
G_{\text{HE}}	= \begin{pmatrix}
	\bar{0} & \bar{1} & \bar{1} & \bar{1} &\bar{0} \\
	\bar{1} & \bar{0} & \bar{1} & \bar{2} &\bar{0} \\ 
	\bar{0} & \bar{0} & \bar{0} & \bar{0} &\bar{1} 		
\end{pmatrix}.
\end{align*}
The next proposition 
shows that the expanded Hamming code is compatible with amalgamated direct sum. 
\begin{proposition} \label{prop:expanded_hamming}
	The expanded Hamming code~$\cC$ has covering radius~$1$. This code is normal, with its last coordinate acceptable.
\end{proposition}
\begin{proof}
    To prove the covering radius claim, we need to show that for any~$\bar{\boldw}\in\bF_3^5$, there exists~$\bar{\boldc}\in\cC$ of Hamming distance at most~$1$ from~$\bar{\boldw}$. If the last entry of~$\bar{\boldw}$ is~$\bar{0}$, then there is some linear combination of the top two rows of $G_{\text{HE}}$ which gives a codeword $\bar{\boldc}\in\cC$ of Hamming distance at most~$1$ from~$\bar{\boldw}$. This follows from the fact that the $[4,2]_3$ Hamming code has covering radius~$1$. If the last entry of~$\bar{\boldw}$ is some~$\bar{a}\ne \bar{0}$, then the last row of $G_{\text{HE}}$ can be used to match the last entry. The claim follows. 

Now we prove that $\cC$ is normal, with its last coordinate acceptable. Let~$\cH$ be the $[4,2]_3$ Hamming code. Observe that for all $\bar{a} \in \bF_3$,
	\begin{align*}
		\cC_{\bar{a}}^{(5)} = \cH\times \{\bar{a} \}. 
	\end{align*} 
	Since the code~$\cH$ is of covering radius~$1$, it follows for any~$\bar{\boldw}\in\bF_3^5$ that
	\begin{align*}
		d_H(\bar{\boldw},\cH\times\{\bar{a}\})\le\begin{cases}
			1 & \mbox{if }\bar{w}_5=\bar{a}\\
			2 & \mbox{else.}
		\end{cases}. 
	\end{align*}
	Therefore,
	\begin{align*}
		N^{(5)}=\max_{\bar{\boldw}\in\bF_3^5}\left\{ \sum_{\bar{a}\in\bF_3}d_H(\bar{\boldw},\cC_{\bar{a}}^{(i)}) \right\}\le 5,
	\end{align*}
	which implies that~$\cC$ is normal and its last coordinate is acceptable.
\end{proof}

Now, since the $i$-repetition code  for~$i\ge 1$ is normal with all coordinates acceptable, it follows that it can be amalgamated with the expanded Hamming code. The resulting code~$\cC_i$ is spanned by the matrix
\begin{align*}
	\begin{pmatrix}
		\bar{0} & \bar{1} & \bar{1} & \bar{1} &\bar{0}^i \\
		\bar{1} & \bar{0} & \bar{1} & \bar{2} &\bar{0}^i \\ 
		\bar{0} & \bar{0} & \bar{0} & \bar{0} &\bar{1}^i 		
	\end{pmatrix},
\end{align*}
where~$\bar{0}^i$ denotes~$i$ occurrences of~$\bar{0}$ (similarly~$\bar{1}^i$), and it is of length~$m=i+4$, size~$27$ and covering radius at most~$r=1+\floor{\frac{2i}{3}}$. It can be verified that~$\tilde{c}=12$, which implies the $\{0, \pm 1\}$-feasibility of the pair $\left\{ \left(\frac{i+16}{i+4},\frac{\floor{\frac{2i}{3}}+2}{i+4}\right) \right\}_{i\ge 1}.$
\paragraph{Comparison with known results}
For $\cA=\{0, \pm1\}$, the best-known approach in the literature is as follows. First construct a $\{\pm 1\}$-protocol using binary covering codes as in \cite{RamkumarRT23_ISIT}. Since $C_2(\{0,\pm1\})=2$,  the result in \cite{RamkumarRT23_Allerton} ensures there is a $\{0, \pm1\}$-protocol with twice the access ratio of the  $\{\pm 1\}$-protocol. It is known \cite{RamkumarRT23_Allerton,RamkumarRT23_ISIT} that points on the line joining two feasible pairs are also feasible; this applies to our results as well. 
Fig.~\ref{fig:comparison} shows the best $\{0,\pm1\}$-feasible pairs given by the ternary codes in the previous subsections. It can be readily seen that our protocols outperform the results from~\cite{RamkumarRT23_Allerton,RamkumarRT23_ISIT}. It follows from Proposition~\ref{proposition:equivalence} that every $3$-element set $\{a,b,c\}$ such that $(b-a)=(c-b)$ has the same feasible pairs as $\{0,\pm1\}$. Therefore, our new points outperform the best-known results for such $3$-element sets as well. 
\begin{figure}[ht!]
  \centering
\includegraphics[width=0.45\textwidth]{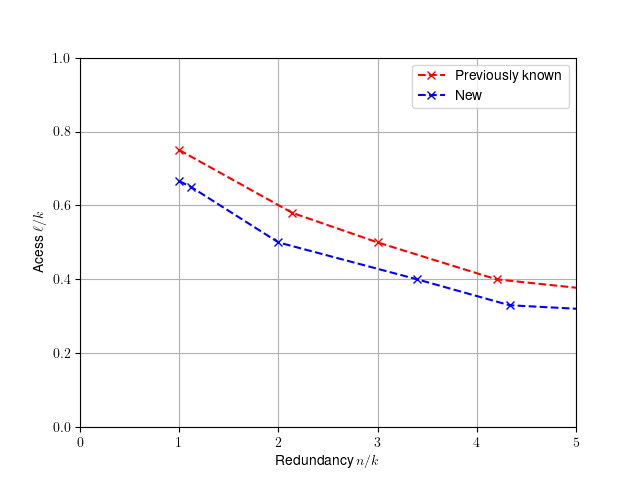}
  \caption{\small{The plot depicts access-redundancy tradeoff for $\cA=\{0, \pm1\}$. The points in blue are the new solutions obtained using ternary covering codes, whereas red is used to indicate the best-known systematic solutions in the literature~\cite{RamkumarRT23_Allerton,RamkumarRT23_ISIT}.}}
  \label{fig:comparison}
\end{figure}

Further, this also leads to better feasible pairs for some larger coefficient sets.  
For instance, consider $\cA=\{0,1,2,3,4,5,6,7,8\}$. As shown in Example~\ref{example:C3}, we have $C_3(\cA)=2$. It can be verified that $C_2(\cA)=4$. If one uses Theorem~\ref{thm:complexity} and the result from \cite{RamkumarRT23_Allerton} to get $\cA$-protocols, the access ratio will be double that of the underlying $\{0, \pm1\}$-protocols in both cases. Therefore, our ternary code method will result in feasible pairs outperforming those given by \cite{RamkumarRT23_Allerton}. 
	
\printbibliography
\clearpage
\appendices
\end{document}